\documentclass[11pt,letterpaper]{article}
\usepackage[tbtags]{amsmath}
\usepackage{amssymb,mathtools,fullpage,color,hyperref,xspace}
\usepackage[all]{hypcap}
\usepackage[amsthm,amsmath,thmmarks]{ntheorem}
\usepackage{graphicx}
\usepackage[margin=1in]{geometry}
\usepackage{array,multirow}

\definecolor{hyper}{RGB}{0,0,88}
\hypersetup{pdfborder={0 0 0},colorlinks=true,linkcolor=hyper,citecolor=hyper,filecolor=hyper,urlcolor=hyper}
\def\equationautorefname~#1\null{(#1)\null}

\newtheorem{theorem}{Theorem}
\newtheorem{corollary}{Corollary}
\newtheorem{lemma}{Lemma}
\newtheorem{claim}{Claim}
\newtheorem{proposition}{Proposition}

\newtheorem{observation}{Observation}

\newtheorem{definition}{Definition}

\newenvironment{mylist}[1]{\begin{itemize}\setlength{\itemsep}{#1pt}\setlength{\parsep}{0pt}\setlength{\parskip}{0pt}}{\end{itemize}}

\newcommand{\newclass}[2]{\newcommand{#1}{{\text{\normalfont\upshape\sffamily #2}}\xspace}}
\newcommand{\renewclass}[2]{\renewcommand{#1}{{\text{\normalfont\upshape\sffamily #2}}\xspace}}

\renewclass{\P}{P}
\newclass{\NP}{NP}
\newclass{\PSPACE}{PSPACE}
\newclass{\NL}{NL}
\renewclass{\L}{L}

\newclass{\LinearT}{Linear Time}

\usepackage{tikz}

 \usepackage[boxruled,linesnumbered, vlined]{algorithm2e}
 \usepackage{algorithmic}

\newcommand{\T}{\text{\normalfont\upshape T}\xspace}
\newcommand{\F}{\text{\normalfont\upshape F}\xspace}

\newcommand{\Fib}{\text{\normalfont\upshape Fib}\xspace}

\usepackage{mathabx}

\usepackage[shortlabels]{enumitem}
\setlist[enumerate]{nosep}

\usepackage{multirow}

\makeatletter
\tikzset{my loop/.style =  {to path={
  \pgfextra{}
  [looseness=5,min distance=5mm]
  \tikz@to@curve@path},font=\sffamily\small
  }}  
\makeatletter 

\renewcommand{\,}{\thinspace\hspace{0pt}}	
\renewcommand{\arraystretch}{1.2}

\SetArgSty{textnormal}

\begin{document}

\title{ Erd\H{o}s--Selfridge Theorem for Nonmonotone CNFs}
\author{Md Lutfar Rahman\qquad Thomas Watson\medskip\\\textit{University of Memphis}}
\maketitle

\begin{abstract}
In an influential paper, Erd\H{o}s and Selfridge introduced the Maker-Breaker game played on a hypergraph, or equivalently, on a monotone CNF. The players take turns assigning values to variables of their choosing, and Breaker's goal is to satisfy the CNF, while Maker's goal is to falsify it. The Erd\H{o}s--Selfridge Theorem says that the least number of clauses in any monotone CNF with $k$ literals per clause where Maker has a winning strategy is $\Theta(2^k)$.

We study the analogous question when the CNF is not necessarily monotone. We prove bounds of $\Theta(\sqrt{2}\,^k)$ when Maker plays last, and $\Omega(1.5^k)$ and $O(r^k)$ when Breaker plays last, where $r=(1+\sqrt{5})/2\approx 1.618$ is the golden ratio.
\end{abstract}

\section{Introduction}
In 1973, Erd\H{o}s and Selfridge published a paper \cite{ES} with several fundamental contributions, including:
\begin{mylist}{2}
\item Being widely regarded as the genesis of the method of conditional expectations. The subsequent impact of this method on theoretical computer science needs no explanation.
\item Introducing the so-called Maker-Breaker game, variants of which have since been studied in numerous papers in the combinatorics literature.
\end{mylist}
We revisit that seminal work and steer it in a new direction. The main theorem from \cite{ES} can be phrased in terms of CNFs (conjunctive normal form boolean formulas) that are monotone (they contain only positive literals). We investigate what happens for general CNFs, which may contain negative literals. We feel that the influence of Erd\H{o}s--Selfridge and the pervasiveness of CNFs in theoretical computer science justify this question as inherently worthy of attention. Our pursuit of the answer uncovers new techniques and invites the development of further techniques to achieve a full resolution in the future.

In the Maker-Breaker game played on a monotone CNF, the eponymous players take turns assigning boolean values to variables of their choosing. Breaker wins if the CNF gets satisfied, and Maker wins otherwise; there are no draws. Since the CNF is monotone, Breaker might as well assign $1$ to every variable she picks, and Maker might as well assign $0$ to every variable he picks. In the generalization to nonmonotone CNFs, each player can pick which remaining variable and which bit to assign it during their turn. To distinguish this general game, we rename Breaker as \T (for ``true'') and Maker as \F (for ``false''). The computational complexity of deciding which player has a winning strategy has been studied in \cite{Sch1,Sch2,Bys,Kutz1,Kutz2,AO,RW1,RW2,RW3}.

A CNF is \emph{$k$-uniform} when every clause has exactly $k$ literals (corresponding to $k$ distinct variables). The Erd\H{o}s--Selfridge Theorem answers an extremal question: How few clauses can there be in a $k$-uniform monotone CNF that Maker can win? It depends a little on which player gets the opening move: $2^k$ if Breaker plays first, and $2^{k-1}$ if Maker plays first. The identity of the player with the final move doesn't affect the answer for monotone CNFs. In contrast, ``who gets the last laugh'' matters a lot for general CNFs:

\begin{theorem}[informal] \label{thm:f-last-informal}
If \F plays last, then the least number of clauses in any $k$-uniform CNF where \F has a winning strategy is $\Theta(\sqrt{2}\,^k)$.
\end{theorem}

\begin{theorem}[informal] \label{thm:t-last-informal}
If \T plays last, then the least number of clauses in any $k$-uniform CNF where \F has a winning strategy is $\Omega(1.5^k)$ and $O(r^k)$ where $r=(1+\sqrt{5})/2\approx 1.618$.
\end{theorem}

The most involved proof is the $\Omega(1.5^k)$ lower bound in \autoref{thm:t-last-informal}. We conjecture the correct bound is $\Theta(r^k)$.

\section{Results}

In the \emph{unordered CNF game}, there is a CNF $\varphi$ and a set of variables $X$ containing all variables that appear in $\varphi$ and possibly more. The players \T and \F alternate turns; each turn consists of picking an unassigned variable from $X$ and picking a value $0$ or $1$ to assign it.\footnote{This game is called ``unordered'' to contrast it with the related TQBF game, in which the variables must be played in a prescribed order.} The game ends when all variables are assigned; \T wins if $\varphi$ is satisfied (every clause has a true literal), and \F wins if $\varphi$ is unsatisfied (some clause has all false literals). There are four possible patterns according to ``who goes first'' and ``who goes last.'' If the same player has the first and last moves, then $|X|$ is odd, and if different players have the first and last moves, then $|X|$ is even.

\begin{definition}
For $k\ge 0$ and $a,b\in\{\T,\F\}$, we let $M_{k,a\cdots b}$ be the minimum number of clauses in $\varphi$, over all unordered CNF game instances $(\varphi,X)$ where $\varphi$ is $k$-uniform and \F has a winning strategy when player $a$ has the first move and player $b$ has the last move.
\end{definition}

\setcounter{theorem}{0}
\begin{theorem}[formal] \label{thm:f-last}
$M_{k,\T\cdots\F}=\sqrt{2}\,^k$ for even $k$, and $1.5\sqrt{2}\,^{k-1}\le M_{k,\T\cdots\F}\le\sqrt{2}\,^{k+1}$ for odd $k$.
\end{theorem}

Let $\Fib_k$ denote the $k^\text{th}$ Fibonacci number. It is well-known that $\Fib_k=\Theta(r^k)$ where $r=(1+\sqrt{5})/2\approx 1.618$.

\begin{theorem}[formal] \label{thm:t-last}
$1.5^k\le M_{k,\T\cdots\T}\le\Fib_{k+2}$ for all $k$.
\end{theorem}

\begin{observation} \label{obs:f-first}
$M_{k,\F\cdots b}=M_{k-1,\T\cdots b}$ for all $k\ge 1$ and $b\in\{\T,\F\}$.
\end{observation}

\begin{proof}
$M_{k,\F\cdots b}\le M_{k-1,\T\cdots b}$: Suppose \F wins $(\varphi,X)$ when \T moves first, where $\varphi$ is $(k-1)$-uniform. Then \F wins $(\varphi',X\cup\{x_0\})$ when \F moves first, where $x_0$ is a fresh variable (not already in $X$) and $\varphi'$ is the same as $\varphi$ but with $x_0$ added to each clause. \F's winning strategy is to play $x_0=0$ first and then use the winning strategy for $(\varphi,X)$. Note that $\varphi'$ is $k$-uniform and has the same number of clauses as $\varphi$.

$M_{k-1,\T\cdots b}\le M_{k,\F\cdots b}$: Suppose \F wins $(\varphi,X)$ when \F moves first, where $\varphi$ is $k$-uniform. Say the opening move in \F's winning strategy is $\ell_i=1$, where $\ell_i\in\{x_i,\overline x_i\}$ is some literal. Obtain $\varphi'$ from $\varphi$ by removing each clause containing $\ell_i$, removing $\overline\ell_i$ from each clause containing $\overline\ell_i$, and removing an arbitrary literal from each clause containing neither $\ell_i$ nor $\overline\ell_i$. Then \F wins $(\varphi',X-\{x_i\})$ when \T moves first, and $\varphi'$ is $(k-1)$-uniform and has at most as many clauses as $\varphi$.
\end{proof}

\begin{corollary} \label{cor:f-first}
~\\[-18pt]
\begin{mylist}{2}
\item $M_{k,\F\cdots\F}=\sqrt{2}\,^{k-1}$ for odd $k$, and $1.5\sqrt{2}\,^{k-2} \le M_{k,\F\cdots\F}\le\sqrt{2}\,^k$ for even $k$.
\item $1.5^{k-1}\le M_{k,\F\cdots\T}\le\Fib_{k+1}$ for all $k$.
\end{mylist}
\end{corollary}

\noindent (\autoref{obs:f-first} requires $k\ge 1$, but the bounds in \autoref{cor:f-first} also hold for $k=0$ since $M_{0,a\cdots b}=1$.)

\section{Upper bounds}
In this section, we prove the upper bounds of \autoref{thm:f-last} and \autoref{thm:t-last} by giving examples of game instances with few clauses where \F wins. In \cite{ES}, Erd\H{o}s and Selfridge proved the upper bound for the Maker-Breaker game by showing a $k$-uniform monotone CNF with $2^k$ clauses where Maker (\F) wins.  The basic idea is that \F can win on the following formula, which is not a CNF: \[ (x_1 \land x_2) \lor (x_3 \land x_4) \lor \cdots \lor (x_{2k-1} \land x_{2k}) \] Whenever \T plays a variable, F responds by assigning $0$ to the paired variable. By the distributive law, this expands to a $k$-uniform monotone CNF with $2^k$ clauses. We study nonmonotone CNFs, which may have both positive and negative literals. 

\subsection{F plays last} \label{ssec:f-last}

\begin{lemma} \label{lem:UETF}
$M_{k,\T \cdots \F} \le \sqrt {2}\,^k$ for even $k$.
\end{lemma}

\begin{proof}
\F can win on the following formula, which is not a CNF, with variables $X_k=\{x_1,\ldots,x_k\}$.
\[ (x_1 \oplus x_2) \lor (x_3 \oplus x_4) \lor \cdots \lor (x_{k-1} \oplus x_k) \]
Whenever \T plays a variable, \F responds by playing the paired variable to make them equal.  To convert this formula to an equivalent CNF, first replace each $(x_i \oplus x_{i+1})$ with $(x_i \lor x_{i+1}) \land (\overline x_i \lor \overline x_{i+1})$. Then by the distributive law, this expands to a $k$-uniform CNF $\varphi_k$ where one clause is \[ ((x_1 \lor x_2) \lor (x_3 \lor x_4) \lor \cdots \lor (x_{k-1} \lor x_k)) \]
and for $i \in \{1,3,5,\ldots, k-1 \}$, each clause contains either $(x_i \lor x_{i+1})$ or $(\overline x_i \lor \overline x_{i+1})$. Therefore $\varphi_k$ has $2^{k/2} = \sqrt {2}\,^k$ clauses: one clause for each $S \subseteq \{1,3,5,\ldots,k-1 \}$. \F wins in $(\varphi_k,X_k)$.
\end{proof}

\begin{lemma} \label{lem:UOTF}
$M_{k,\T \cdots \F} \le \sqrt {2}\,^{k+1}$ for odd $k$.
\end{lemma}

\begin{proof}
Suppose $\varphi_{k-1}$ is the $(k-1)$-uniform CNF with $\sqrt {2}\,^{k-1}$ clauses from \autoref{lem:UETF} (since $k-1$ is even).  We take two copies of $\varphi_{k-1}$, and put a new variable $x_k$ in each clause of one copy, and a new variable $x_{k+1}$ in each clause of the other copy. Call this $\varphi_k$. Formally:
\begin{align*}
& \varphi_{k}  = \bigwedge \limits_{ C\in \varphi_{k-1}} (C \lor x_{k}  ) \land ( C\lor x_{k+1} )  \\
& X_{k}  = \{ x_1, x_2, \ldots, x_{k+1} \} 
\end{align*}
We argue \F wins in $(\varphi_{k}, X_k)$. If \T plays $x_k$ or $x_{k+1}$, \F responds by assigning $0$ to the other one. For other variables, \F follows his winning strategy for $(\varphi_{k-1}, X_{k-1})$ from \autoref{lem:UETF}.  Since $\varphi_{k-1}$ is a $(k-1)$-uniform CNF with $\sqrt {2}\,^{k-1}$ clauses, $\varphi_{k}$ is a $k$-uniform CNF with $2\sqrt {2}\,^{k-1} = \sqrt {2}\,^{k+1}$ clauses.
\end{proof}

\subsection{T plays last}

Before proving \autoref{lem:UTT} we draw an intuition. We already know that \F wins on 
\[ (x_1 \land x_2) \lor (x_3 \land x_4) \lor \cdots \lor (x_{2k-1} \land x_{2k}). \] 
Now replace each $(x_i \land x_{i+1})$ with $(x_i \land (\overline x_i \lor x_{i+1}))$, which is equivalent. This does not change the function expressed by the formula, so \F still wins this \T$\cdots$\F game. To turn it into a \T$\cdots$\T game, we can introduce a dummy variable $x_0$. Since the game is equivalent to a monotone game, neither player has any incentive to play $x_0$, so \F still wins this \T$\cdots$\T game.

If we convert it to a CNF, then by the distributive law it will again have $2^k$ clauses. But this CNF is not uniform---each clause has at least $k$ literals and at most $2k$ literals. We can do a similar construction that balances the CNF to make it uniform. This intuitively suggests that $\sqrt {2}\,^k < M_{k,\T \cdots \T} < 2^k$. 

\begin{lemma} \label{lem:UTT}
$M_{k,\T \cdots \T} \le\Fib_{k+2} $.
\end{lemma}
\begin{proof} 
For every $k \in \{0,1,2, \ldots\}$ we recursively define a $k$-uniform CNF $\varphi_k$ on variables $X_k$, where $X_k = \{ x_0,x_1, \ldots, x_{2k-2} \}$ if $k>0$, and $X_0=\{ x_0\}$  (these $\varphi_k$, $X_{k}$ are different than in \autoref{ssec:f-last}):
\begin{mylist}{2}
\item $k=0$: $\varphi_0 = () $
\item $k=1$: $\varphi_1 = (x_0) \land (\overline x_0) $
\item $k>1$: $ \varphi_k = \bigwedge \limits_{ C \in \varphi_{k-1}} (C \lor x_{2k-3}) \land   \bigwedge \limits_{ C \in \varphi_{k-2}} (C \lor \overline x_{2k-3} \lor x_{2k-2} ) $
\end{mylist}
Now we argue \F wins in $(\varphi_k, X_{k})$. \F's strategy is to assign $0$ to at least one variable from each pair  $\{ x_1,x_2\}, \{x_3,x_4\}, \{x_5,x_6\}, \ldots, \{x_{2k-3}, x_{2k-2}\}$. Whenever \T plays from a pair, \F responds by assigning $0$ to the other variable. After \T plays $x_0$, \F picks a fresh pair $\{x_i,x_{i+1}\}$ where $i$ is odd and assigns one of them $0$, then ``chases'' \T until \T plays the other from $\{x_i,x_{i+1}\}$. Here the 	``chase'' means whenever \T plays from a fresh pair, \F responds by assigning $0$ to the other variable in that pair.  After \T returns to $\{x_i,x_{i+1}\}$, then \F picks another fresh pair to start another chase, and so on in phases.  We prove by induction on $k$ that this strategy ensures $\varphi_k$ is unsatisfied:
\begin{mylist}{2}
\item $k=0$: $\varphi_0$ is obviously unsatisfied.
\item $k=1$: $\varphi_1$ is obviously unsatisfied.
\item $k>1$:  By induction, both $\varphi_{k-1}$ and $\varphi_{k-2}$ are unsatisfied. Now $\varphi_k$ is unsatisfied since: By \F's strategy, at least one of $\{x_{2k-3}, x_{2k-2}\}$ is assigned $0$. If $x_{2k-3}=0$ then one of the clauses of $\varphi_k$ that came from $\varphi_{k-1}$ is unsatisfied. If $x_{2k-3}=1$ and $x_{2k-2}=0$ then one of the clauses of $\varphi_k$ that came from $\varphi_{k-2}$ is unsatisfied.
\end{mylist}
Letting $|\varphi_k|$ represent the number of clauses in $\varphi_k$, we argue $|\varphi_k| = \Fib_{k+2}$ by induction on $k$:
\begin{mylist}{2}
\item $k=0$: $|\varphi_0| = 1 = \Fib_2$.
\item $k=1$: $|\varphi_1| = 2 = \Fib_3$.
\item $k>1$:  By induction, $|\varphi_{k-1}| = \Fib_{k+1}$ and $|\varphi_{k-2}| = \Fib_{k}$.  So \[ |\varphi_k| = |\varphi_{k-1}| + |\varphi_{k-2}| = \Fib_{k+1} + \Fib_{k} = \Fib_{k+2}. \]
\end{mylist}
Therefore $M_{k,\T \cdots \T} \le\Fib_{k+2}$.
\end{proof}
\section{Lower bounds}
\subsection{Notation}
In the proofs, we will define a potential value $p(C)$ for each clause $C$. The value of $p(C)$ depends on the context. If $\varphi$ is a CNF (any set of clauses), then the potential of $\varphi$ is $p(\varphi) = \sum_{C \in \varphi } p(C)$. The potential of a literal $\ell_i$ with respect to $\varphi$ is defined as $p(\varphi, \ell_i) = p(\{C \in \varphi : \ell_i \in C\})$. When we have a particular $\varphi$ in mind, we can abbreviate $p(\varphi, \ell_i)$ as $p(\ell_i)$.

Suppose $\varphi$ is a CNF and $\ell_i, \ell_j$ are two literals. We define the potentials of different sets of clauses based on which of $\ell_i, \ell_j$, and their complements exist in the clause. For example, $a(\varphi, \ell_i, \ell_j)$ is the sum of the potentials of clauses in $\varphi$ that contain both $\ell_i, \ell_j$.
\begin{center}
	\setlength\tabcolsep{4pt}
{\renewcommand{\arraystretch}{1.3}
\begin{tabular}{r|>{\centering\arraybackslash}m{2.7cm}|>{\centering\arraybackslash}m{2.7cm}|>{\centering\arraybackslash}m{2.7cm}|}
\multicolumn{1}{r}{}&\multicolumn{1}{c}{$\ell_j$}&\multicolumn{1}{c}{$\overline\ell_j$}&\multicolumn{1}{c}{neither $\ell_j$ nor $\overline\ell_j$}\\
\cline{2-4}$\ell_i$&$a$&$b$&$c$\\
\cline{2-4}$\overline\ell_i$&$d$&$e$&$f$\\
\cline{2-4}neither $\ell_i$ nor $\overline\ell_i$&$g$&$h$&\\
\cline{2-4}
\end{tabular}}
\end{center}

\begin{align*}
a(\varphi, \ell_i, \ell_j)~&=~p(\{C \in \varphi : \ell_i \in C ~\text{and}  ~\ell_j \in C\}) \\
b(\varphi, \ell_i, \ell_j)~&=~p(\{C \in \varphi : \ell_i \in C ~\text{and} ~\overline \ell_j \in C\}) \\
c(\varphi, \ell_i, \ell_j)~&=~p(\{C \in \varphi : \ell_i \in C ~\text{and} ~ \ell_j \notin C ~\text{and} ~\overline \ell_j \notin C \}) \\
d(\varphi, \ell_i, \ell_j)~&=~p(\{C \in \varphi : \overline \ell_i \in C ~\text{and} ~ \ell_j \in C\}) \\
e(\varphi, \ell_i, \ell_j)~&=~p(\{C \in \varphi : \overline \ell_i \in C ~\text{and} ~ \overline \ell_j \in C\}) \\
f(\varphi, \ell_i, \ell_j)~&=~p(\{C \in \varphi : \overline \ell_i \in C ~\text{and} ~ \ell_j \notin C ~\text{and} ~\overline \ell_j \notin C \}) \\
g(\varphi, \ell_i, \ell_j)~&=~p(\{C \in \varphi : \ell_i \notin C ~\text{and} ~ \overline \ell_i \notin C ~\text{and} ~ \ell_j \in C\}) \\
 h(\varphi, \ell_i, \ell_j)~&=~p(\{C \in \varphi : \ell_i \notin C ~\text{and} ~ \overline \ell_i \notin C ~\text{and} ~\overline \ell_j \in C\}) 
\end{align*}
We can abbreviate these quantities as $a,b,c,d,e,f,g,h$ in contexts where we have particular $\varphi, \ell_i, \ell_j$ in mind. Also the following relations hold:
\begin{align*}
p(\ell_i)~&=~a+b+c \\
p(\overline \ell_i)~&=~d+e+f \\
p(\ell_j)~&= ~a+d+g \\
p(\overline \ell_j)~&=~b+e+h 
\end{align*}

When we assign $\ell_i=1$ (i.e., assign $x_i=1$ if $\ell_i$ is $x_i$, or assign $x_i=0$ if $\ell_i$ is $\overline x_i$), $\varphi$ becomes the \emph{residual} CNF denoted $\varphi[\ell_i=1]$ where all clauses containing $\ell_i$ get removed, and the literal $\overline\ell_i$ gets removed from remaining clauses.

\subsection{F plays last}

\begin{lemma} \label{lem: 2nTT}
$M_{k,\T \cdots \F} \ge \sqrt {2}\,^k$ for even $k$.
\end{lemma}

\begin{proof}
Consider any $\T\cdots\F$ game instance $(\varphi,X)$ where $\varphi$ is a $k$-uniform CNF with $<\sqrt{2}\,^k$ clauses and $|X|$ is even. We show \T has a winning strategy. In this proof, we use $p(C)=1/\sqrt{2}\,^{|C|}$. A \emph{round} consists of a \T move followed by an \F move. 

\begin{claim} \label{clm:main}
In every round, there exists a move for \T such that for every response by \F, we have $p(\psi)\ge p(\psi')$ where $\psi$ is the residual CNF before the round and $\psi'$ is the residual CNF after the round.
\end{claim}

At the beginning we have $p(C)=1/\sqrt{2}\,^k$ for each clause $C\in\varphi$, so $p(\varphi)<\sqrt{2}\,^k/\sqrt{2}\,^k=1$. By \autoref{clm:main}, \T has a strategy guaranteeing that $p(\psi)\le p(\varphi)<1$ where $\psi$ is the residual CNF after all variables have been played. If this final $\psi$ contained a clause, the clause would be empty and have potential $1/\sqrt{2}\,^0=1$, which would imply $p(\psi)\ge 1$. Thus the final $\psi$ must have no clauses, which means $\varphi$ got satisfied and \T won.
This concludes the proof of \autoref{lem: 2nTT}, except for the proof of \autoref{clm:main}.
\end{proof}
\begin{proof}[Proof of \autoref{clm:main}]
Let $\psi$ be the residual CNF at the beginning of a round. \T picks a literal $\ell_i$ maximizing $p(\psi,\ell_i)$ and plays $\ell_i=1$.\footnote{It is perhaps counterintuitive that \T's strategy ignores the effect of clauses that contain $\overline\ell_i$, which increase in potential after playing $\ell_i=1$. A more intuitive strategy would be to pick a literal $\ell_i$ maximizing $p(\psi,\ell_i)-(\sqrt{2}-1)p(\psi,\overline\ell_i)$, which is the overall decrease in potential from playing $\ell_i=1$; this strategy also works but is trickier to analyze.} Suppose \F responds by playing $\ell_j=1$, and let $\psi'$ be the residual CNF after \F's move. Letting the $a,b,c,d,e,f,g,h$ notation be with respect to $\psi,\ell_i,\ell_j$, we have \[p(\psi)-p(\psi')=a+b+c+d+g- \bigl ( e+(\sqrt{2}-1)(f+h) \bigr ) \] because:
\begin{mylist}{2}
\item Clauses from the $a,b,c,d,g$ groups are satisfied and removed (since they contain $\ell_i=1$ or $\ell_j=1$ or both), so their potential gets multiplied by $0$.
\item Clauses from the $e$ group each shrink by two literals (since they contain $\overline\ell_i=0$ and $\overline\ell_j=0$), so their potential gets multiplied by $\sqrt{2}\cdot\sqrt{2}=2$.
\item Clauses from the $f,h$ groups each shrink by one literal, so their potential gets multiplied by $\sqrt{2}$.
\end{mylist}
By the choice of $\ell_i$, we have $p(\ell_i)\ge p(\overline\ell_i)$ and $p(\ell_i)\ge p(\overline\ell_j)$ with respect to $\psi$, in other words, $a+b+c\ge d+e+f$ and $a+b+c\ge b+e+h$. Thus $p(\psi)\ge p(\psi')$ because
\begin{align*}
a+b+c+d+g~\ge~a+b+c ~ \textstyle\ge~\frac{1}{2}(d+e+f)+\frac{1}{2}(b+e+h) & ~\textstyle\ge~e+\frac{1}{2}(f+h) \\
& ~\ge~e+(\sqrt{2}-1)(f+h).
\end{align*}
\end{proof}
Note: It did not matter whether $k$ is even or odd! \autoref{lem: 2nTT} is true for any $k$. \autoref{lem:tfodd} actually uses oddness of $k$. The main idea is to exploit the slack $1/2 \ge\sqrt{2}-1$ that appeared at the end of the proof of \autoref{clm:main}.

\begin{lemma} \label{lem:tfodd}
 $M_{k,\T \cdots \F} \ge 1.5\,\sqrt 2\,^{k-1}$ for odd $k$.
\end{lemma}

\begin{proof}
Consider any $\T\cdots\F$ game instance $(\varphi,X)$ where $\varphi$ is a $k$-uniform CNF with $<1.5\,\sqrt 2\,^{k-1}$ clauses and $|X|$ is even. We show \T has a winning strategy. In this proof, we use 
\begin{align*}
  p(C)=\begin{cases}
    1/\sqrt 2\, ^ {|C|} & \text{if $|C|$ is even}.\\
    1/1.5  \sqrt 2\, ^ {|C|-1} & \text{if $|C|$ is odd}.
  \end{cases}
\end{align*} 

\begin{claim} \label{clm:main2}
In every round, there exists a move for \T such that for every response by \F, we have $p(\psi)\ge p(\psi')$ where $\psi$ is the residual CNF before the round and $\psi'$ is the residual CNF after the round.
\end{claim}

At the beginning we have $p(C)= 1/1.5  \sqrt 2\, ^ {k-1}$ for each clause $C\in\varphi$ (since $|C|=k$, which is odd), so $p(\varphi)< 1.5\,\sqrt 2\,^{k-1}/1.5\,\sqrt 2\,^{k-1}=1$. By \autoref{clm:main2}, \T has a strategy guaranteeing that $p(\psi)\le p(\varphi)<1$ where $\psi$ is the residual CNF after all variables have been played. If this final $\psi$ contained a clause, the clause would be empty and have potential $1/\sqrt{2}\,^0=1$ (since $0$ is even), which would imply $p(\psi)\ge 1$. Thus the final $\psi$ must have no clauses, which means $\varphi$ got satisfied and \T won. This concludes the proof of \autoref{lem:tfodd}, except for the proof of \autoref{clm:main2}.
\end{proof}

\begin{proof}[Proof of \autoref{clm:main2}]
Let $\psi$ be the residual CNF at the beginning of a round. \T picks a literal $\ell_i$ maximizing $p(\psi,\ell_i)$ and plays $\ell_i=1$. Suppose \F responds by playing $\ell_j=1$, and let $\psi'$ be the residual CNF after \F's move. Letting the $a,b,c,d,e,f,g,h$ notation be with respect to $\psi,\ell_i,\ell_j$, we have \[p(\psi)-p(\psi') \ge a+b+c+d+g- \bigl ( e+\textstyle{\frac{1}{2}}(f+h) \bigr )\] because:
\begin{mylist}{2}
\item Clauses from the $a,b,c,d,g$ groups are satisfied and removed (since they contain $\ell_i=1$ or $\ell_j=1$ or both), so their potential gets multiplied by $0$.
\item Clauses from the $e$ group each shrink by two literals (since they contain $\overline\ell_i=0$ and $\overline\ell_j=0$). Here odd-width clauses remain odd and even-width clauses remain even, so their potential gets multiplied by $\sqrt{2}\cdot\sqrt{2}=2$.
\item Clauses from the $f,h$ groups each shrink by one literal.
There are two cases for a clause $C$ in these groups: 
\begin{mylist}{2}
\item $|C|$ is even, so $p(C) =  1/\sqrt 2  \, ^ {|C|}$. After $C$ being shrunk by 1, the new clause $C'$ has potential $ p(C') =  1/1.5  \sqrt 2 \, ^ {|C'|-1} = 1/1.5  \sqrt 2  \, ^ {|C|-2}$. So the potential of an even-width clause gets multiplied by $p(C')/p(C) = 4/3$. 
\item $|C|$ is odd, so $p(C) =  1/1.5 \sqrt 2  \, ^ {|C|-1}$. After $C$ being shrunk by 1, the new clause $C'$ has potential $ p(C') =  1/\sqrt 2  \, ^ {|C'|} =1/\sqrt 2  \, ^ {|C|-1} $. So the potential of an odd-width clause gets multiplied by  $p(C')/p(C) = 3/2$.
\end{mylist}
So their potential gets multiplied by $\le 3/2$ (since $4/3 \le 3/2$).
\end{mylist}
By the choice of $\ell_i$, we have $p(\ell_i)\ge p(\overline\ell_i)$ and $p(\ell_i)\ge p(\overline\ell_j)$ with respect to $\psi$, in other words, $a+b+c\ge d+e+f$ and $a+b+c\ge b+e+h$. Thus $p(\psi)\ge p(\psi')$ because
\[\textstyle a+b+c+d+g~\ge~a+b+c~\ge~\frac{1}{2}(d+e+f)+\frac{1}{2}(b+e+h)~\ge~e+\frac{1}{2}(f+h).\]
\end{proof}

\subsection{T plays last}
\begin{lemma} \label{lem: LTT}
$M_{k,\T \cdots \T} \ge 1.5\,^k$.
\end{lemma}

\begin{proof}
Consider any $\T\cdots\T$ game instance $(\varphi,X)$ where $\varphi$ is a $k$-uniform CNF with $<1.5^k$ clauses and $|X|$ is odd. We show \T has a winning strategy. In this proof, we use $p(C)=1/1.5^{|C|}$.

For intuition, how can \T take advantage of having the last move? She will look out for certain pairs of literals to ``set aside'' and wait for \F to assign one of them, and then respond by assigning the other one the opposite value. We call such a pair ``zugzwang,'' which means a situation where \F's obligation to make a move is a disadvantage for \F. Upon finding such a pair, \T anticipates that certain clauses will get satisfied later, but other clauses containing those literals might shrink when the zugzwang pair eventually gets played. Thus \T can update the CNF to pretend those events have already transpired. The normal gameplay of $\T\F$ rounds (\T plays, then \F plays) will sometimes get interrupted by $\F\T$ rounds of playing previously-designated zugzwang pairs. We define the zugzwang condition so that \T's modifications won't increase the potential of the CNF (which is no longer simply a residual version of $\varphi$). When there are no remaining zugzwang pairs to set aside, we can exploit this fact---together with \T's choice of ``best'' literal for her normal move---to analyze the potential change in a $\T\F$ round. This allows the proof to handle a smaller potential function and hence more initial clauses, compared to when \F had the last move.

\LinesNumberedHidden{
\begin{algorithm}[p] 
\caption{\T's winning strategy in $(\varphi,X)$} \label{alg:tt}
\DontPrintSemicolon
\BlankLine
 \SetKwProg{Fn}{subroutine}{:}{}

initialize ~$\psi \leftarrow \varphi$;~~$Y \leftarrow X$;~~$\zeta \leftarrow \{\}$;~~$Z \leftarrow \{\}$ \;

\While{game is not over}{

	\While{{\normalfont{\texttt{FindZugzwang()}}} returns a pair ($\ell_i,\ell_j$)}{
		\texttt{TfoundZugzwang($\ell_i,\ell_j$)}
	}
	\texttt{TplayNormal()}\;
	\While{\F picks $x_k \in Z$ and $\ell_k \in \{x_k, \overline x_k\}$ and assigns $\ell_k=1$}{
		\texttt{TplayZugzwang($\ell_k$)}
	}
	\lIf{$|Y \cup Z|=0$}{halt}
	\texttt{FplayNormal()}
}
 \;

 \SetKwFunction{FindZugzwang}{FindZugzwang}
  \Fn{\FindZugzwang{}}{
       \lIf {there exist distinct $x_i,x_j\in Y$ and $\ell_i \in \{x_i,\overline x_i\}$ and $\ell_j\in\{x_j,\overline x_j\}$ such that (with respect to $\psi, \ell_i, \ell_j$): $\textstyle a+e\ge\frac{5}{4}(b+d)+\frac{1}{2}(c+f+g+h)$}{\KwRet $(\ell_i, \ell_j)$} 
      \KwRet  NULL \;
  }
\;
 \SetKwFunction{TfoundZugzwang}{TfoundZugzwang}
  \Fn{\TfoundZugzwang{$\ell_i,\ell_j$}}{
     /* \T modifies $\psi$ with the intention to make $\ell_i \ne \ell_j$ by waiting for \F to touch $\{x_i,x_j\}$ */ \;
      $\zeta\leftarrow\zeta\cup\{(\ell_i\oplus\ell_j)\}$;~~$Z \leftarrow Z \cup \{ x_i, x_j\}$;~~$Y \leftarrow Y - \{ x_i, x_j\}$  \;
       remove from $\psi$ every clause containing $\ell_i \lor \ell_j$ or containing $\overline \ell_i \lor \overline \ell_j$\;
       remove $\ell_i, \overline \ell_i,  \ell_j, \overline \ell_j$ from all other clauses of $\psi$
  }
\;
 \SetKwFunction{TplayZugzwang}{TplayZugzwang}
  \Fn{\TplayZugzwang{$\ell_k$}}{
  /* \T makes $\ell_m \ne \ell_k$ */ \;
\T picks $x_m\in Z$ and $\ell_m\in\{x_m,\overline x_m\}$ such that $(\ell_k\oplus\ell_m)\in\zeta$ and assigns $\ell_m=0$\;
   $\zeta \leftarrow \zeta  - \{(\ell_k \oplus \ell_m)\} $;~~$Z \leftarrow Z - \{ x_k, x_m\}$  
  }
\;
 \SetKwFunction{Tplay}{TplayNormal}
  \Fn{\Tplay{}}{
\T picks $x_i\in Y$ and $\ell_i\in\{x_i,\overline x_i\}$ maximizing $p(\psi,\ell_i)-p(\psi,\overline\ell_i)$ and assigns $\ell_i=1$\;
$\psi\leftarrow\psi[\ell_i=1]$;~~$Y\leftarrow Y-\{x_i\}$
  }
\;
\SetKwFunction{Fplay}{FplayNormal}
  \Fn{\Fplay{}}{
\F picks $x_j\in Y$ and $\ell_j\in\{x_j,\overline x_j\}$ and assigns $\ell_j=1$\;
$\psi\leftarrow\psi[\ell_j=1]$;~~$Y\leftarrow Y-\{x_j\}$; 
  }

\end{algorithm}
}

We describe \T's winning strategy in $(\varphi,X)$ as \autoref{alg:tt}. In the first line, the algorithm declares and initializes $\psi, Y, \zeta, Z$, which are accessed globally. Here $\psi$ is a CNF (initially the same as $\varphi$), and $\zeta$ is a set (conjunction) of constraints of the form $(\ell_i\oplus\ell_j)$. We consider $(\ell_i\oplus\ell_j)$, $(\ell_j\oplus\ell_i)$, $(\overline\ell_i\oplus\overline\ell_j)$, $(\overline\ell_j\oplus\overline\ell_i)$ to be the same constraint as each other.
The algorithm maintains the following three invariants:
\begin{mylist}{2}
\item[(1)] $Y$ and $Z$ are disjoint subsets of $X$, and $Y\cup Z$ is the set of unplayed variables, and $Y$ contains all variables that appear in $\psi$, and $Z$ is exactly the set of variables that appear in $\zeta$, and $|Z|$ is even.
\item[(2)] For every assignment to $Y \cup Z$, if $\psi$ and $\zeta$ are satisfied, then $\varphi$ is also satisfied by the same assignment together with the assignment played by \T and \F so far to the other variables of $X$.
\item[(3)] $p(\psi) < 1$.
\end{mylist}
Now we argue how these invariants are maintained at the end of the outer loop in \autoref{alg:tt}. Invariant (1) is straightforward to see.

\begin{claim}
 Invariant $(2)$ is maintained.
\end{claim}
\begin{proof}
Invariant (2) trivially holds at the beginning.

Each iteration of the first inner loop maintains (2): Say $\psi$ and $\zeta$ are at the beginning of the iteration, and $\psi'$ and $\zeta'$ denote the formulas after the iteration. Assume (2) holds for $\psi$ and $\zeta$. To see that (2) holds for $\psi'$ and $\zeta'$, consider any assignment to the unplayed variables. We will argue that if $\psi'$ and $\zeta'$ are satisfied, then $\psi$ and $\zeta$ are satisfied, which implies (by assumption) that $\varphi$ is satisfied. So suppose $\psi'$ and $\zeta'$ are satisfied. Then $\psi$ is satisfied because each clause containing $\ell_i \lor \ell_j$ or containing $\overline \ell_i \lor \overline \ell_j$ is satisfied due to $(\ell_i \oplus \ell_j)$ being satisfied in $\zeta'$, and each other clause is satisfied since it contains the corresponding clause in $\psi'$ which is satisfied. Also, $\zeta$ is satisfied since each of its constraints is also in $\zeta'$ which is satisfied.

It is immediate that \T's and \F's ``normal'' moves in the outer loop maintain (2), because of the way we update $\psi$ and $Y$.

Each iteration of the second inner loop maintains (2): If an assignment satisfies $\psi'$ and $\zeta'$ (after the iteration) then it also satisfies $\psi$ and $\zeta$ (at the beginning of the iteration) since \T's move satisfies $(\ell_k\oplus\ell_m)$---and therefore the assignment satisfies $\varphi$.
\end{proof}

\begin{claim} \label{clm:i3}
 Invariant $(3)$ is maintained.
\end{claim}

\begin{proof}
Invariant (3) holds at the beginning by the assumption that $\varphi$ has $<1.5^k$ clauses (and each clause has potential $1/1.5^{k}$).

The first inner loop maintains (3) by the following proposition, which we prove later.

\begin{proposition} \label{fact:zug}
If {\normalfont{\texttt{FindZugzwang()}}} returns $(\ell_i, \ell_j)$, then $p(\psi)\ge p(\psi')$ where $\psi$ and $\psi'$ are the CNFs before and after the execution of  {\normalfont{\texttt{TfoundZugzwang()}}}.
\end{proposition}

The second inner loop does not affect (3). In each outer iteration except the last, \T's and \F's moves from $Y$ maintain (3) by the following proposition, which we prove later.

\begin{proposition} \label{fact:zug2}
If {\normalfont{\texttt{FindZugzwang()}}} returns {\normalfont NULL}, then $p(\psi)\ge p(\psi')$ where $\psi$ is the CNF before {\normalfont{\texttt{TplayNormal()}}} and $\psi'$ is the CNF after {\normalfont{\texttt{FplayNormal()}}}.
\end{proposition}

This concludes the proof of \autoref{clm:i3}.
\end{proof}

Now we argue why \T wins in the last outer iteration. Right before {\normalfont{\texttt{TplayNormal()}}}, $|Y|$ must be odd by invariant (1), because an even number of variables have been played so far (since \T has the first move) and $|X|$ is odd (since \T also has the last move) and $|Z|$ is even. Thus, \T always has an available move in {\normalfont{\texttt{TplayNormal()}}} since $|Y|>0$ at this point. When \T is about to play the last variable $x_i \in Y$ (possibly followed by some $Z$ moves in the second inner loop), all remaining clauses in $\psi$ have width $\le 1$. There cannot be an empty clause in $\psi$, because then $p(\psi)$ would be $ \ge 1/1.5^0 = 1$, contradicting invariant (3). There cannot be more than one clause in $\psi$, because then $p(\psi)$ would be $\ge 2/1.5^{1} \ge 1$. Thus $\psi$ is either empty (already satisfied) or just $(x_i)$ or just $(\overline x_i)$, which \T satisfies in one move.

At termination, $Y$ and $Z$ are empty, and $\psi$ and $\zeta$ are empty and thus satisfied. By invariant (2), this means $\varphi$ is satisfied by the gameplay, so \T wins.

This concludes the proof of \autoref{lem: LTT} except \autoref{fact:zug} and \autoref{fact:zug2}.
\end{proof}

\begin{proof}[Proof of \autoref{fact:zug}] \label{p1}
Since {\normalfont{\texttt{FindZugzwang()}}} returns $(\ell_i, \ell_j)$, the following holds with respect to $\psi, \ell_i, \ell_j$:
\begin{equation} \label{eq:wait} \tag{$\spadesuit$}
\textstyle a+e\ge\frac{5}{4}(b+d)+\frac{1}{2}(c+f+g+h)  
\end{equation}
We also have
\[ \textstyle p(\psi)-p(\psi')=a+e - \bigl ( \frac{5}{4}(b+d)+\frac{1}{2}(c+f+g+h) \bigr ) \] because:
\begin{mylist}{2}
\item Clauses from the $a,e$ groups are removed (since they contain $\ell_i \lor \ell_j$ or $\overline \ell_i \lor \overline \ell_j$), so their potential gets multiplied by $0$. (Intuitively, \T considers these clauses satisfied in advance since she will satisfy $(\ell_i\oplus\ell_j)$ later.)

\item Clauses from the $b,d$ groups each shrink by two literals (since they contain two of $\ell_i, \overline \ell_i, \ell_j, \overline \ell_j$ which are removed), so their potential gets multiplied by $1.5 \cdot 1.5 = 9/4$. (Some of these four literals will eventually get assigned $1$, but since \T cannot predict which ones, she pessimistically assumes they are all $0$.)

\item Clauses from the $c,f,g,h$ groups each shrink by one literal (since they contain one of $\ell_i, \overline \ell_i, \ell_j, \overline \ell_j$ which are removed), so their potential gets multiplied by $1.5 = 3/2$.
\end{mylist}
Since \autoref{eq:wait} holds, $p(\psi) \ge p(\psi')$. 
\end{proof}

\begin{proof}[Proof of \autoref{fact:zug2}] \label{p2}
In \texttt{TplayNormal()}, \T picks the literal $\ell_i$ maximizing $p(\psi,\ell_i)-p(\psi,\overline\ell_i)$ and plays $\ell_i=1$.\footnote{Some other strategies would also work here, but this one is the simplest to analyze.} In \texttt{FplayNormal()}, \F plays $\ell_j=1$. With respect to $\psi, \ell_i, \ell_j$ we have \[ \textstyle p(\psi)-p(\psi')=a+b+c+d+g- \bigl ( \frac{5}{4}e+\frac{1}{2}(f+h) \bigr ) \] because:
\begin{mylist}{2}
\item Clauses from the $a,b,c,d,g$ groups are satisfied and removed (since they contain $\ell_i=1$ or $\ell_j=1$ or both), so their potential gets multiplied by $0$.
\item Clauses from the $e$ group each shrink by two literals (since they contain $\overline\ell_i=0$ and $\overline\ell_j=0$), so their potential gets multiplied by $1.5 \cdot 1.5 = 9/4$.
\item Clauses from the $f,h$ groups each shrink by one literal, so their potential gets multiplied by $1.5 = 3/2$.
\end{mylist}
By the choice of $\ell_i$ (i.e., maximizing  $p(\ell_i)-p(\overline\ell_i)$), we have:
\begin{equation} \label{eq:1} \tag{$\clubsuit$}
\begin{split}
&p(\ell_i)-p(\overline\ell_i) \ge p(\overline\ell_j)-p(\ell_j)\\
&\implies a+b+c-d-e-f \ge b+e+h-a-d-g \\
&\implies 2a +0b + 1c + 0d - 2e - 1f + 1g -1h \ge 0
\end{split}
\end{equation}
Since \texttt{FindZugzwang()} returns NULL, \autoref{eq:wait} does not hold in $\psi$. Thus the following holds:
\begin{equation} \label{eq:2} \tag{$\blacklozenge$}
\begin{split}
&\textstyle (a+e) < \frac{5}{4}(b+d) + \frac{1}{2}(c+f+g+h) \\
& \textstyle \implies -1a + \frac{5}{4}b + \frac{1}{2}c + \frac{5}{4}d - 1e + \frac{1}{2}f + \frac{1}{2}g + \frac{1}{2}h > 0
\end{split}
\end{equation}
Thus $p(\psi)\ge p(\psi')$ because the linear combination $\frac{9}{16}$\autoref{eq:1} $+ \frac{1}{8}$\autoref{eq:2} implies:
\begin{align*}
&\textstyle \frac{9}{16}\bigl(2a +0b + 1c + 0d - 2e - 1f + 1g -1h\bigr) +  \frac{1}{8}\bigl(-1a + \frac{5}{4}b + \frac{1}{2}c + \frac{5}{4}d - 1e + \frac{1}{2}f + \frac{1}{2}g + \frac{1}{2}h\bigr) > 0\\
& \textstyle \implies 1a + \frac{5}{32}b + \frac{5}{8}c +  \frac{5}{32}d - \frac{5}{4}e - \frac{1}{2}f  + \frac{5}{8}g - \frac{1}{2}h > 0\\
&\textstyle \implies 1a + 1b + 1c + 1d - \frac{5}{4}e - \frac{1}{2}f + 1g - \frac{1}{2}h > 0 \\
&\textstyle \implies a+b+c+d+g - \bigl ( \frac{5}{4}e + \frac{1}{2}(f+h) \bigr ) > 0
\end{align*}
\end{proof}

\subsection*{Acknowledgments} 
This work was supported by NSF grants CCF-1657377 and CCF-1942742.


\bibliographystyle{alphaurl}
\bibliography{est}

\end{document}